\documentclass[11pt]{article}
\usepackage[english]{babel}
\usepackage[letterpaper,margin=1in]{geometry}
\usepackage{amsmath}
\usepackage{graphicx}
\usepackage[colorlinks=true, allcolors=blue]{hyperref}
\usepackage{xcolor}
\definecolor{MyPurple}{RGB}{111,0,255}

\usepackage[utf8]{inputenc} 
\usepackage{url}            
\usepackage{booktabs}       
\usepackage{amsfonts}       
\usepackage{nicefrac}       
\usepackage{microtype}      
\usepackage{xcolor}

\usepackage{bbm}
\usepackage{amsthm}
\usepackage{amsmath}        
\usepackage{amssymb}        
\usepackage{comment}        
\usepackage{enumitem}  
\usepackage{parskip}
\usepackage{tikz}
\usepackage{subcaption}
\usetikzlibrary{calc}
\newtheorem{theorem}{Theorem}[section]
\newtheorem*{theorem*}{Theorem}
\newtheorem{proposition}[theorem]{Proposition}

\newtheorem*{remark*}{Remark}
\newtheorem{lemma}[theorem]{Lemma}

\newcommand{\be}{\boldsymbol e}

\newcommand{\Acal}{\mathcal{A}}

\newcommand{\Lcal}{\mathcal{L}}

\newcommand{\Ocal}{\mathcal{O}}
\newcommand{\Pcal}{\mathcal{P}}

\newcommand{\Tcal}{\mathcal{T}}

\newcommand{\Vcal}{\mathcal{V}}

\newcommand{\Nbb}{\mathbb{N}}

\newcommand{\Rbb}{\mathbb{R}}
\newcommand{\R}{\mathbb{R}}

\newcommand{\bx}{{\boldsymbol x}}

\newcommand{\blambda}{{\boldsymbol \lambda}}
\newcommand{\bgamma}{{\boldsymbol \gamma}}
\newcommand{\bxi}{{\boldsymbol \xi}}

\newcommand{\one}{\mathbbm{1}}

\newcommand{\OT}{\text{OT}}
\newcommand{\Cost}{{\normalfont \texttt{Cost}}}

\DeclareMathOperator*{\argmin}{arg\,min}

\title{Unweighted Layered Graph Traversal:\\Passing a Crown via Entropy Maximization}
\author{Xingjian Bai \\ MIT \and Christian Coester \\ University of Oxford \and Romain Cosson \\ Inria, Paris}

\begin{document}

\date{}
\maketitle

\begin{abstract}
Introduced by Papadimitriou and Yannakakis in 1989, layered graph traversal is a central problem in online algorithms and mobile computing that has been studied for several decades, and which now is essentially resolved in its original formulation. In this paper, we demonstrate that what appears to be an innocuous modification of the problem actually leads to a drastic (exponential) reduction of the competitive ratio. Specifically, we present an algorithm that is $\Ocal(\log^2 w)$-competitive for traversing unweighted layered graphs of width $w$. Our algorithm chooses the agent’s position simply according to the probability distribution over the current layer that maximizes the sum of entropies of the induced distributions in the preceding layers.
\end{abstract}


\thispagestyle{empty}
\clearpage
\pagenumbering{arabic}

\section{Introduction}
Exploring an unknown environment with a mobile agent is a fundamental task in robotics and artificial intelligence \cite{papadimitriou1991shortest}. Despite its importance, very few models allow for rigorous worst-case analysis, especially when the environment is a general graph (or metric space). One successful model introduced in the field of online algorithms is called `layered graph traversal', and it has been studied extensively since the 1990s \cite{papadimitriou1991shortest,chrobak1993metrical,ramesh1995traversing,burley1996traversing,fiat1998competitive,bubeck2022shortest,bubeck2023randomized, banerjee2023graph,cosson2024collective}. In this paper, we show that a simple and natural assumption on the problem -- specifically, that the unknown graph is unweighted -- drastically reduces its competitive ratio. 

\paragraph{Problem setting.} In layered graph traversal, a mobile agent starts from some node -- called the source -- of an unknown weighted graph $G=(V, E)$, and is tasked to reach some other node -- called the target. The graph is divided into ``layers'', where the $i$-th layer refers to the set of nodes at combinatorial depth $i$ (i.e. $i$ hops away from the source).  The agent is short-sighted in the sense that it only gets to see layer $i+1$ once it is located at layer $i$. Luckily, the agent is broad-sighted, in the sense that it sees all nodes and edges going from layer $i$ to layer $i + 1$ when it reaches layer $i$. The problem is parameterized by the \textit{width} $w$ of the graph, defined as the maximum number of nodes in a layer. 
The cost of the agent is defined as the total distance traveled until reaching the target. A deterministic (resp. randomized) algorithm for layered graph traversal is said to be $c_w$-competitive if its cost (resp. expected cost) is at most $c_w$ times the length of the shortest path from source to target. The deterministic competitive ratio of layered graph traversal is known to lie between $\Omega(2^w)$ \cite{fiat1998competitive} and $\Ocal(w 2^w)$ \cite{burley1996traversing} and its randomized counterpart was settled to $\Theta(w^2)$ in \cite{bubeck2022shortest,bubeck2023randomized}. It is known~\cite{fiat1998competitive} that layered graph traversal is equivalent to the special case where the graph is a tree and all edge lengths are either $0$ or $1$.

In the unweighted variant (where all edge lengths are exactly 1), a depth-first search approach (DFS) is clearly $2w-1$-competitive (see Appendix \ref{ap:simpleResults}). Despite the broad interest that layered graph traversal received, no randomized algorithm was shown to improve over this bound until \cite{cosson2024collective} recently proposed a $\Ocal(\sqrt{w})$-competitive randomized algorithm in the unweighted setting.

\begin{figure}[h!]
  \centering
  
  \includegraphics[width=0.5\textwidth]{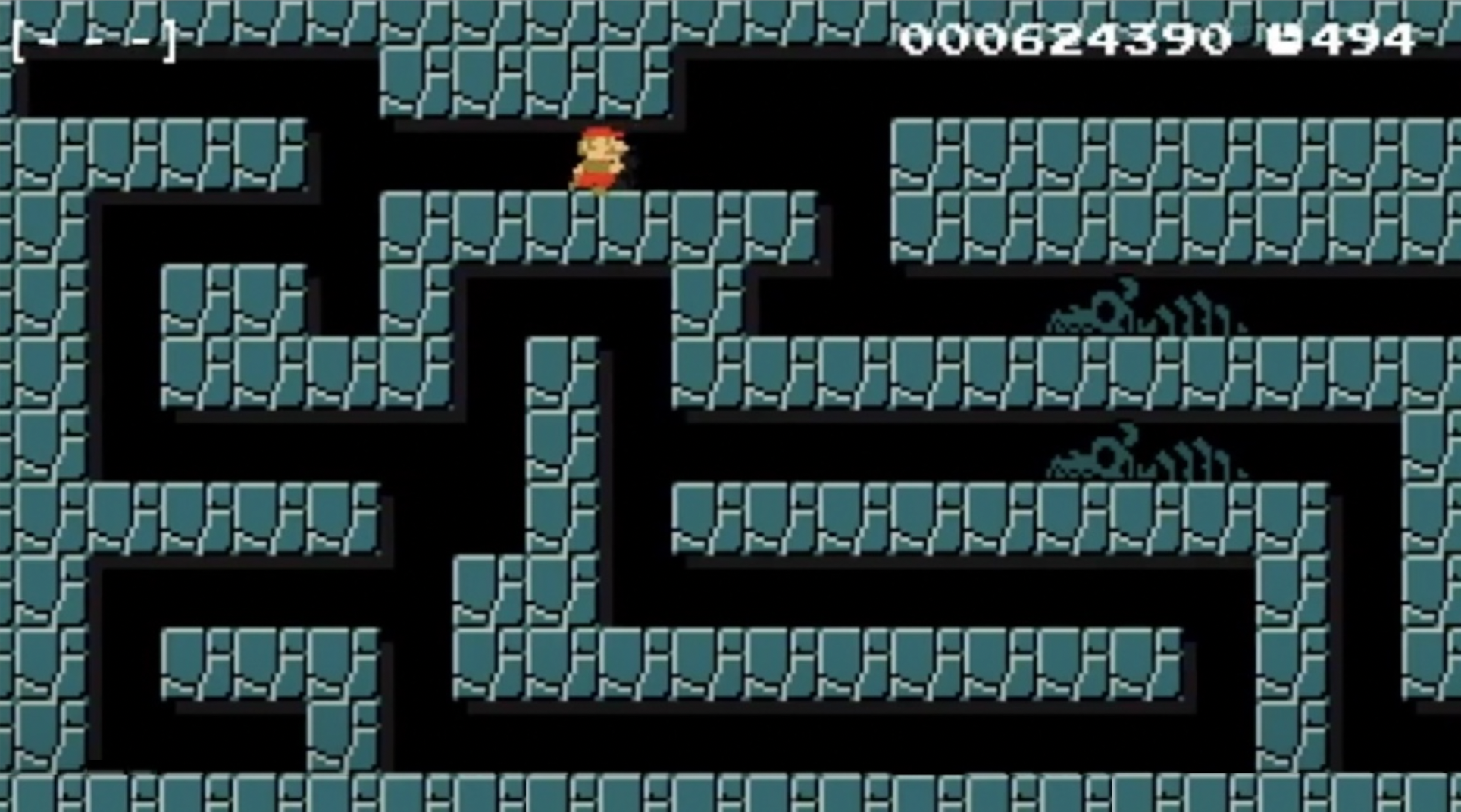}
  
  \caption{Mario is a well-known example of an agent that is short-sighted but broad-sighted. The width $w$ corresponds to the screen height, with the target located somewhere off to the right. A natural question is: faced with the unknown maze ahead, which path should Mario choose?}
  \label{fig:example}
\end{figure}

\paragraph{Our results.} In this paper, we show that the assumption that the graph is unweighted allows for an even more drastic improvement in the competitive ratio, reducing it to $\Ocal(\log^2 w)$. Our algorithm is remarkably
simple: It maintains at all times that the probability distribution of the agent's location maximizes 
an entropy defined on the revealed part of the graph. The majority of this paper is dedicated to proving this randomized upper bound, based on extending techniques from the online mirror descent framework~\cite{BubeckCLLM18,BubeckCLL21} with several new ideas. A few simple motivating results, such as tightness of $2w-1$ for deterministic algorithms (achieved by depth-first search), a randomized lower bound of $\Omega(\log w)$, and some natural but unsuccessful randomized algorithms are surveyed in Appendix~\ref{ap:simpleResults}.

\paragraph{Applications: playing Mario, passing a crown and learning-augmented algorithms.} An immediate application of unweighted layered graph traversal is illustrated in Figure \ref{fig:example}, where the famous video game character Mario is short-sighted and broad-sighted by design. 

On a more general ground, there are many situations where an agent is \emph{short-sighted in time} (because it cannot predict the future) but \emph{broad-sighted in space} (because it sees all past and current states of the world). 
For an amusing example, consider the following application: suppose a king or queen must design a rule for the passing of their crown and wishes to minimize how much it will ``travel'' between their descendants (supposedly, less movement leads to more political stability). They could decide that the leader of each new generation shall be chosen by an unweighted layered graph traversal algorithm. Note that the traditional method ``give the crown to the closest living relative of the last holder'' is equivalent to depth-first search, which has lesser competitive guarantees than the algorithm we introduce. Admittedly, the interest of today's monarchies in deviating from traditional methods might be limited.

The principle of passing a crown/baton applies to wider settings, such as combining strategies in the context of learning-augmented algorithms: \cite{AntoniadisCEPS23b} recently proposed to use a layered graph traversal algorithm to choose between $w$ expert strategies, when there is a (metric) cost associated to switching between experts. In the general setting (when the experts move at arbitrary speeds) a weighted layered graph traversal algorithm induces a competitive ratio in 
$O(w^2)$, but in settings where all expert strategies move at the same speed, an unweighted layered graph traversal algorithm could be employed, resulting in a competitive ratio in 
$O(\log^2 w)$ against the best combination of expert strategies.

\subsection{Discussion of techniques}
\paragraph{Entropy maximization.} The downside of deterministic online algorithms is that they are completely predictable. This enables strong lower bounds, as an adversary can construct the input in the way that is most detrimental to the choices of the algorithm. Therefore, the key to achieving better competitive ratios via randomization lies in making algorithms less predictable. The motivating idea of our algorithm is to take this strategy to the extreme, by aiming for the \emph{most unpredictable} configuration at all times. The unpredictability of a random decision is typically measured by its entropy, $\sum_u x_u\log \frac{1}{x_u}$, where $x_u$ denotes the probability of decision $u$.

To implement the idea of maximizing unpredictability, a first tempting idea might be to choose the probability distribution over nodes $\ell$ of the current layer that maximizes $\sum_\ell x_\ell\log \frac{1}{x_\ell}$, i.e., the uniform distribution. However, this ignores that the distribution over the current layer also induces distributions over all previous layers; for example, if all but one node of the current layer share the same ancestor in the preceding layer, then the uniform (high-entropy) distribution over the current layer would induce a very non-uniform (low-entropy) distribution over their ancestors in the preceding layer. We wish to maximize not only the entropy of the current layer, but also of their ancestors in each previous layer. Hence, our algorithm chooses the probability distribution over the current layer that maximizes the \emph{sum} of entropies over \emph{all} layers.

\paragraph{Relation to prior mirror descent based online algorithms.} Entropy-driven online algorithms have led to major advances recently~\cite{BuchbinderCN14,BubeckCLLM18,BuchbinderGMN19,BubeckCLL21,CoesterL22,BansalC22,bubeck2022shortest}. Unlike our algorithm, these prior algorithms are not expressed as straightforwardly as merely maximizing entropy, but are formulated in the framework of online mirror descent. The entropic functionals (called \emph{regularizers}) in these prior works involve a perturbation term $\delta>0$, e.g. $\Phi(x)=\sum (x_u+\delta)\log(x_u+\delta)$ in its simplest forms, and an algorithm (in the continuous-time framework of~\cite{BubeckCLLM18,BubeckCLL21}) is then described by a differential inclusion
\begin{align*}
    \nabla^2 \Phi(\bx)\dot\bx -\gamma  \in -N_{K},
\end{align*}
where $\dot\bx$ is the time derivative of configuration $\bx$, $\gamma$ is the ``desired direction of movement'', and $N_K$ is the normal cone of the polytope $K$ of configurations. The interpretation of $\Phi$ in these cases is that it replaces the Euclidean geometry (which would correspond to the case where $\nabla^2\Phi(x)$ is replaced by the identity matrix, as in gradient descent) by a geometry adapted to the uncertainty of the underlying problem. Although our algorithm also satisfies such an inclusion, it can be formulated even more simply as
\begin{align*}
    x(i) = \argmin_{x\in K}\Phi_i(x),
\end{align*}
where $\Phi_i(x)$ denotes the sum of negative entropies of all layers visible at time $i$. Unlike previous mirror descent based online algorithms, this formulation in fact yields a \emph{memoryless} algorithm since the selected $\Phi_i$ only depends on the tree distances between the nodes of the current layer. This also allows us to obtain an alternative expression of our algorithm through elementary operations (see Lemma \ref{lem:explicitAlgo}).

\paragraph{Achieving $\boldsymbol{\Ocal(\log^2w)}$.} The analysis of our algorithm is split into deactivation phases (where nodes that become irrelevant are deactivated) and growth phases (where new nodes are added), which at a high level resemble similar phases previously considered in the ``evolving tree game'' of~\cite{bubeck2022shortest}. However, our methods of analysing these phases have little in common with~\cite{bubeck2022shortest} and require several novel ideas.

Our regularizer can be equivalently written as
\begin{align}
\label{eq: expression_for_phi}
    \Phi_i(x)=\sum_u x_u\log x_u = \sum_u h_{p(u)}x_u\log\frac{x_u}{x_{p(u)}},
\end{align}
where the sum is over vertices $u$ of the Steiner tree of the current layer, $x_u$ is the algorithm's probability of residing in the subtree rooted at $u$, $p(u)$ is the parent of $u$, and $h_{u}$ denotes the round number $i$ minus the combinatorial depth of $u$ in the tree, i.e. the height of $u$ at round $i$. Note that this quantity is well-defined only because the tree is unweighted, and all nodes at the last layer are at the same combinatorial depth. 
See Appendix~\ref{ap: useful-identities} for the identity between the two sums.

The right-hand side of equation~\eqref{eq: expression_for_phi} reminds us of the ``conditional entropy'' regularizer employed in~\cite{CoesterL22} for metric task systems on hierarchically separated trees. 
This expression allows us to derive an expression for the change of the conditional probabilities $\frac{x_u}{x_{p(u)}}$ in each step (which does not involve Lagrange multipliers). However, the analysis in~\cite{CoesterL22} requires that edge lengths decrease geometrically along root-to-leaf paths, and more crucially it only achieves a competitive ratio \emph{linear in the depth} of the tree. In our case, the depth of the tree (even if merging edges as in~\cite{bubeck2022shortest}) can be $\Omega(w)$, so an analysis similar to~\cite{CoesterL22} would fail to achieve a sublinear competitive ratio. The high depth of the tree is also the obstacle leading to the $\Theta(w^2)$ competitive ratio  for weighted layered graph traversal in~\cite{bubeck2022shortest}. To achieve the exponential improvement to $\Ocal(\log^2w)$, a key idea of our proof is to separate the analysis of movement cost when handling deactivation phases, where the probability mass $x_\ell$ of a leaf $\ell$ decreases to $0$, into two substeps that are analyzed differently depending on whether $x_\ell$ is smaller or greater than $1/\text{poly}(w)$. These two substeps require two different potential functions, which we combine into an overall potential function that is compatible with both substeps of deactivation phases and also with growth phases.

Another interesting aspect -- which echoes a technique in \cite{cosson2024collective} --
is that the value of the regularizer appears in the potential function, instead of the Bregman divergence with an offline configuration.
The Bregman divergence would be problematic for us since, due to the lack of $\delta$-perturbations in $\Phi_i$, its Lipschitz constant can be $\Omega(w)$, even when evaluated only at the relevant online algorithm configurations. 

\subsection{Related work} The setting of layered graph traversal was introduced in a paper titled `Shortest path without a map' by Papadimitriou and Yannakakis, motivated by applications in navigation robotics \cite{papadimitriou1991shortest}. For the case of $w=2$, they showed that the deterministic variant of the problem is essentially equivalent to linear search, a problem which dates back to \cite{bellman1963optimal} for which the competitive ratio is equal to $9$ \cite{baezayates1993searching}. The first competitive analysis for arbitrary values of $w\in \Nbb$ was obtained by \cite{fiat1998competitive}, which coined the term `layered graph traversal' to emphasize how the agent collects information (i.e., the feedback). Following a series of improvements \cite{fiat1998competitive,ramesh1995traversing,burley1996traversing}, the competitive ratio of the deterministic variant of the problem was proved to lie between $\Ocal(w2^w)$ \cite{burley1996traversing} and $\Omega(2^w)$ \cite{fiat1998competitive}. More recently, the competitive ratio of the randomized variant of the problem was resolved up to a constant factor, establishing it at $\Theta(w^2)$ \cite{bubeck2022shortest,bubeck2023randomized}. For the unweighted setting, \cite{cosson2024collective} recently provided a $\Ocal(\sqrt{w})$-competitive algorithm, leaving a substantial gap compared to the classical $\Omega(\log w)$ lower-bound. This paper significantly narrows this gap, presenting a $\Ocal(\log^2 w)$-competitive algorithm.

The problem was also studied under the name `metrical service systems' \cite{chrobak1993metrical} to highlight its connection to the setting of metrical task systems \cite{borodin1992optimal}. 
The denomination `small set chasing' was also recently suggested by \cite{bubeck2022shortest} to point out its connection with more recent works on set chasing (see, e.g. \cite{bubeck2019competitively,argue2021chasing,Sellke20}). The problem is additionally related to several other problems in online computation, such as the $k$-taxi problem~\cite{CoesterK19}, whose lower bound is based on a reduction from layered graph traversal. Very recently, layered graph traversal has also found new applications in the context of learning-augmented algorithms~\cite{AntoniadisCEPS23a,AntoniadisCEPS23b}.

\textit{Exploration of an unknown environment by a mobile agent.} In layered graph traversal, the agent's objective can be restated as follows: the agent must build a complete map of the underlying graph using the fewest moves possible. 
Graph exploration with a mobile agent has attracted attention since the 19th century \cite{lucas1882recreations,tarry1895probleme}, marked by the formal introduction of the depth-first search (DFS) as a maze-solving algorithm. 
DFS is particularly suited for exploring general unknown graphs when the agent is short-sighted and narrow-sighted: the agent traverses all $m$ edges in exactly $2m$ moves, which is optimal in the sense of the competitive ratio. A natural extension of this problem is when the agent can also see its immediate neighborhood (a little less short-sighted) \cite{kalyanasundaram1994constructing}. This setting is known as `online graph exploration' \cite{megow2012online}, and the goal is for the agent to visit all nodes while incurring a limited cost with respect to the optimal tour (i.e., the solution of the traveling salesman problem). The competitive ratio for that problem is still unsettled (it lies between $\Omega(1)$ and $\Ocal(\log n)$, where $n$ is the number of nodes in the graph) but many restricted classes of graphs admit $\Ocal(1)$-competitive algorithms, such as planar graphs or some minor-excluded graphs (see \cite{baligacs2023exploration} and its references). In contrast, layered graph traversal proposes a different model of feedback where the agent gets to see simultaneously all nodes and edges at some given hop distance from its initial position. This setting is particularly well-motivated when the graph evolves with time (e.g., a phylogenetic tree). Many other graph exploration settings have been introduced, varying the capabilities of the agent: e.g., with limited memory \cite{aleliunas1979random}, limited storage \cite{yanovski2003distributed}, or considering multiple agents working a as a team \cite{fraigniaud2006collective,cosson2023efficient,cosson2024collective}. For a recent review of graph exploration problems, see the introduction of \cite{cosson2024ariadne}. Contrarily to the aforementioned examples, this paper studies the case of a single, fully-competent agent.

\subsection{Notations and preliminaries} 
\label{sec:notations}
Consider a graph $G=(V, E)$ defined by a set of nodes $V\subset \Vcal$ and edges $E$ over $V$, where $\Vcal$ is an arbitrary countable set. 
The edges are undirected and unweighted, and the graph is connected. One node is called the root and is denoted by $r\in V$. For some $i\in \Nbb$, we denote by $\Lcal(i)$ the set of nodes at combinatorial depth $i$, i.e., $i$ hops away from the root. We will refer to $\Lcal(i)$ as the $i$-th layer of the graph. 
The width $w$ of the graph is defined as the maximum cardinality of its layers, i.e., $w = \max_{i} |\Lcal(i)|$. As standard in layered graph traversal, the goal to reach a certain target can be equivalently simplified to reaching the last layer of the graph. 

\textit{Reduction to trees.} For any node $u \in V$, we denote by $i_u$ the combinatorial depth of $u$, i.e. satisfying $u\in \Lcal(i_u)$. At any step $i\geq i_u$, we call $h_u(i) = i-i_u$ the height of $u$. For any node $u\in V\setminus \{r\}$, we can arbitrarily pick a node of the preceding layer $p(u) \in \Lcal(i_u-1)$, such that $(p(u),u)\in E$.\footnote{Observe that the choice of $p(u)$ can be made online by the agent. The reduction from graphs to tres is standard in the literature on layered graph traversal. Here, we use the assumption that any non-root node is connected to a node in the preceding layer because a layer refers to the set of nodes at the same combinatorial depth. We note that this assumption is not required in the weighted variant of the problem; see \cite{fiat1998competitive}. For more motivation on the setting of this paper, we refer to \cite{cosson2024collective}.} It is clear that traversing the tree formed by $T=\{(p(u),u)\mid u\in V\setminus\{r\}\} \subseteq E$ in an online manner is no less difficult than traversing the original graph. For this reason, we will focus the rest of the paper on the specific case where the layered graph is a tree. For any node $u\in V$, we denote by $c(u)\subset V$ the set of its children of $u$, and by $\Lcal_u(i)$ the subset of nodes of $\Lcal(i)$ that are descendants of $u$, for some $i\in \Nbb$. We also denote by $u\rightarrow r \subset V$ the set of all nodes in the path from $u$ to $r$, excluding $r$.

\textit{Layered graph traversal algorithm.} Formally, a randomized unweighted layered graph traversal algorithm is defined by a function  $\Acal: \Tcal \times \Vcal \rightarrow \Pcal(\Vcal)$ mapping a tree $T \in \Tcal$ and a node $u\in \Vcal$ in the penultimate layer of $T$ to a distribution $\bx \in \Pcal(\Vcal)$ over the set of nodes in the ultimate layer of $T$.
The algorithm is deterministic if the distribution is always supported on a unique node. The cost of the algorithm is defined as the total number of edges traversed by an agent that picks its next position using $\Acal$ at each layer and uses the shortest paths to move between positions elected in consecutive layers. The total cost of an agent until reaching layer $i$ depends on the random choices taken by the agent (if the algorithm is randomized) and we therefore focus on its expectation which we denote by $\Cost(i)$. The algorithm is $c_w$-competitive if, for any graph of width $w$ and any step $i$, it satisfies $\Cost(i)\leq c_w i$. 

\textit{Fractional perspective on unweighted layered graph traversal.} A classical equivalent perspective on layered graph traversal -- which is referred to as the fractional view -- is that the algorithm $\Acal:\Tcal \rightarrow \Pcal(\Vcal)$ maintains a distribution on the ultimate layer of the tree given as input. The movement cost charged to the algorithm is then the optimal transport cost between the two consecutive distributions. Given a fractional algorithm for layered graph traversal, one gets a randomized algorithm with the same cost by considering at each step an optimal coupling between the two consecutive distributions and sampling the next position of the agent according to the conditional distribution associated with its current position. We refer to \cite[Section 4]{bubeck2022shortest} for more details on the reduction. 

\textit{The active tree and the active polytope.}
At any round $i$, we say that a node is `active' if it has a descendant in $\Lcal(i)$ and we denote the set of all active nodes by $V(i)$. A node is `deactivated' at round $i$ if it ceases to be active at that round and the set of all deactivated nodes until round $i$ is denoted by $\bar V(i)$. We will essentially focus on the `active tree' that is formed by the nodes in $V(i)$. Observe that this tree is alternatively defined as the Steiner tree of $\Lcal(i)\cup\{r\}$.
We will denote by $\bx(i) \in \Pcal(\Lcal(i))$ the probability distribution associated with the position of the agent on the $i$-th layer. Instead of viewing $\bx$ as a probability distribution, it will be useful to consider it more generally as a point of the \textit{active polytope}, $K(i)$ that we define as,
\begin{equation*}
    K(i) = \{\bx \in \Rbb^{V(i)\setminus\{r\}}\mid \forall u \in V(i) \setminus \Lcal(i)\colon x_u = \sum_{v\in c(u)}x_v \}. 
\end{equation*}
We emphasize that $x_r$ is defined once and for all to be $1$. For that reason, it does not appear as a variable in the definition of the active polytope above (but $x_r$ does appear in one constraint, which effectively enforces $\sum_{\ell\in \Lcal(i)}x_\ell(i) =1$). We note that a probability distribution must also satisfy non-negativity constraints, which are not explicit in $K(i)$. It will be obvious that the configuration $\bx(i)\in K(i)$ defined by our algorithm satisfies the positivity constraints (with strict inequality) and can indeed be interpreted as $\bx(i) \in \Pcal(\Lcal(i))$. 
We note that $K(i)$ is an affine subspace of $\R^{V(i)\setminus\{r\}}$, 
Denoting by \(\mathbf{e}_u \in \mathbb{R}^{V(i)\setminus\{r\}}\) the canonical basis vector associated with the element \(u\in V(i)\setminus\{r\}\), 
we define the normal cone $N_{K(i)}$ associated with $K(i)$, 
\begin{align*}
    N_{K(i)} &= \text{Vect}(\{\be_{u}-\sum_{v\in c(u)}\be_v: u\in V(i) \setminus \Lcal(i) \}),\\
    &= \{\bxi\in \Rbb^{V(i) \setminus \{r\}} : \xi_u = \lambda_{p(u)}-\lambda_u  \text{ with } \blambda \in \Rbb^{V(i)} \text{ and } \lambda_\ell = 0 \text{ for } \ell\in \Lcal(i)\},
\end{align*}
which is the vector space orthogonal to the the affine subspace $K(i)$.

\textit{Optimal transport in a tree.} For two given configurations, $\bx(i) \in K(i)$ and $\bx(i + 1)\in K(i + 1)$, we denote the optimal transport cost (also known as the earth-mover distance) between $\bx(i)$ and $\bx(i + 1)$ by $\OT(\bx(i),\bx(i + 1))$. 
Extending the definition of a configuration to take value zero outside of the active polytope, the optimal transport cost on a tree is classically expressed as
\begin{equation*}
    \OT(\bx,\bx') = \sum_{u\in V}|x_u'-x_u|.
\end{equation*}
We note that this definition matches the usual definition of optimal transport on the metric space associated with the underlying tree. 

\textit{Derivatives and instantaneous movement costs.} The derivative of a differentiable real function $f(\cdot)$ with respect to a single variable $t\in \Rbb$ will be denoted by $\partial_t f(t)$. The gradient of a twice-differentiable multivariate function $\Phi(\bx)$ will be denoted by $\nabla \Phi(\bx)$, and its Hessian will be denoted by $\nabla^2 \Phi(\bx)$. For convenience, when the function that is differentiated is a configuration $\bx(t) \in K(i)$, we will adopt the physicist's notation and denote its derivative $\partial_t \bx(t)$ by $\dot \bx$. Thus, at an arbitrary node $u\in V(i)$, we will also denote $\partial_t x_u(t)$ by $\dot x_u$. In general, we will often drop dependence on time when it is clear from the context. We observe that if $\forall t: \bx(t)\in K(i)$, then $\dot \bx(t)$ is orthogonal to the normal cone $N_{K(i)}$, i.e., $\dot \bx(t)\cdot \bxi=0$ for $\bxi\in N_{K(i)}$. We define the instantaneous movement cost at instant $t$ as 
\begin{equation}
    \partial_t \Cost(t) = \sum_{u\in V(i)}|\dot x_u|.\label{def: instant-movement}
\end{equation}
By the triangle inequality, it is clear that for any two configurations $\bx\in K(i)$ and $\bx'\in K(i)$, if $\bx(t)$ is a differentiable function with respect to $t$, interpolating between $\bx(a)=\bx$ and $\bx(b)=\bx'$ for $a<b$, we have $\OT(\bx,\bx') \leq \int_a^b \partial_t \Cost(t)dt$. Overloading notation, we will use letter $i$ to refer to the actual configurations $\bx(i)$ taken by our algorithm at discrete time steps, and letter $t$ to denote continuously evolving configurations $\bx(t)$ interpolating between two consecutive configurations.

\textit{Convexity.} For some twice differentiable convex function $\Phi$ defined on an open domain $U$, we denote the Bregman divergence associated to $\Phi$ by $D(\cdot,\cdot)$. It is defined for any $\bx, \bx' \in U^2$ as, $D(\bx',\bx) = \Phi(\bx')-\Phi(\bx) -\nabla \Phi(\bx) (\bx'-\bx)$, and it is always non-negative.

\section{Algorithm and analysis}
In this section, we present and analyze our algorithm for layered graph traversal, which is stated here in its fractional formulation (see Section~\ref{sec:notations} for definition). When reaching layer $i\in\Nbb$, the algorithm chooses the configuration $\bx(i)$ that is the minimizer of the following expression, 
\begin{equation}
\label{alg:definition}
    \bx(i) \in \argmin_{\bx \in K(i)} \sum_{u\in V} x_u \log x_u,
\end{equation}
which 
we denote by $\Phi_i$.\footnote{We note that $\Phi_i$ is defined on the open domain $U = \Rbb_{++}^{V(i)}$ and that the minimizer is effectively taken on  $K(i) \cap U$. Such minimizer exists due to the property that $\partial_{x_u}\phi(\bx) \rightarrow -\infty$ when $x_u\rightarrow 0$. It is also unique, by strong convexity of $\Phi_i$.} 
We observe that $\Phi_i$ can equivalently be rewritten as
\begin{equation*}
\Phi_i(\bx)  = \sum_{u\in V(i)} x_u \log x_u = \sum_{u\in V(i)\setminus \{r\}}h_{p(u)}(i)x_u\log\frac{x_u}{x_{p(u)}}, 
\end{equation*}
where $h_{p(u)}(i)$ denotes the height of $p(u)$ at round $i$. See Appendix \ref{ap: useful-identities} for details. The rest of this section is devoted to proving the following theorem, which is the main result of the paper. 
\begin{theorem} The fractional algorithm defined by \eqref{alg:definition} is $\Ocal(\log^2 w)$-competitive for unweighted layered graph traversal. In other words, the associated randomized algorithm satisfies for any unweighted graph of width $w$ that its expected cost satisfies at all rounds $i$, 
$$\Cost(i) \leq \Ocal(\log^2 w)\cdot i.$$
\end{theorem}

\textit{Proof sketch.}
For any layer $i\in \Nbb$, we define the potential
\begin{equation*}
    P(i) = \frac{4}{w}|\bar V(i)| + 4\log w \cdot  \Phi_i(\bx(i)) + 6(1+\log w)^2i, 
\end{equation*}
where $|\bar V(i)|$ denotes the number of nodes that have been deactivated by round $i$ (see Section~\ref{sec:notations}).
We will prove that the cost of the fractional algorithm up to round $i$ -- which we denote by $\Cost(i)$ -- is always bounded by the potential, 
\begin{equation}\label{eq:CleP}
    \Cost(i) \leq  P(i).
\end{equation}
Since $|\bar V(i)| \leq w i$ and $\Phi_i(\bx(i))\leq 0$, this implies that $\Cost(i) \leq \Ocal(\log^2 w)\cdot i$.

The proof is by induction on $i$. When moving from layer $i$ to layer $i+1$, the goal is to show that the increase in cost, which equals $\OT(\bx(i),\bx(i+1))$, is at most the variation of the potential, which is equal to $P(i+1)-P(i)$. Instead of considering the optimal transport from $\bx(i)\in K(i)$ to $\bx(i+1)\in K(i+1)$, we consider a series of (possibly sub-optimal) phases that together form a transport between $\bx(i)$ and $\bx(i+1)$. There are two main phases.

\begin{enumerate}
    \item For each node $\ell\in \Lcal(i)$ that does not have a descendant in $\Lcal(i+1)$, we displace the fractional mass present at $\ell$ to obtain $x_\ell = 0$. We call this phase the \textit{deactivation phase}. The deactivation of leaf $\ell$ is performed continuously by setting $\bx(t) = \argmin_{x\in K(i)} \Phi_{i}(\bx)+t x_\ell$ and letting $t\rightarrow\infty$. After the deactivation of $\ell$, in slight abuse of notation, we remove $\ell$ from the set $V(i)$ and write $K(i)$ for the new corresponding polytope. 
    After consecutively performing all deactivations, we obtain a new configuration, $\bx'$, and re-interpret it as a point in $K(i+1)$ by setting $x_\ell=x_{p(\ell)}/|c(p(\ell))|$ for each $\ell\in\Lcal(i+1)$. 
    At the end of the deactivation phase, we have $\bx' \in \argmin_{\bx \in K(i+1)} \Phi_{i}(\bx)$. 
    The deactivation phase is studied in Section \ref{sec:deactivate} where we show the following bound on the associated movement cost:
\begin{equation}\label{eq:main-delete}
        \OT(\bx(i),\bx') \leq \frac{4}{w}\left(|\bar V(i+1)|-|\bar V(i)|\right) + 4\log w \cdot\left(\Phi_i(\bx')-\Phi_i(\bx(i))\right).
    \end{equation}

    \item Then, we consider the configurations $\bx(t) = \argmin_{\bx \in K(i+1)} \Phi_{i}(\bx)+t\sum_{\ell \in \Lcal(i+1)}x_\ell \log x_\ell$ for values of $t$ that increase continuously from $0$ to $1$, which interpolate between $\bx'$ and $\bx(i+1)$. We call this the \textit{growth phase}.
    After the growth, we add a cost of $1$ to account for the ultimate movement from $\Lcal(i)$ to $\Lcal(i+1)$. The growth phase is studied in Section \ref{sec:grow}, where we will see that 
    \begin{align}
    \OT(\bx',\bx(i+1))&\leq 2(1+\log w)^2+1,\label{eq:main-elong-1}\\
    -\log w&\leq \Phi_{i+1}(\bx(i+1))-\Phi_i(\bx').\label{eq:main-elong-2} 
    \end{align}
    
\end{enumerate}
Combining the above equations as $\eqref{eq:main-delete}+\eqref{eq:main-elong-1}+4\log w \cdot \eqref{eq:main-elong-2}$ yields
\begin{equation*}
    \OT(\bx(i),\bx(i+1)) \leq P(i+1) - P(i),
\end{equation*}
which in turn implies \eqref{eq:CleP} by induction. Inequalities \eqref{eq:main-delete}, \eqref{eq:main-elong-1} and \eqref{eq:main-elong-2} also help to understand the the design of the potential $P(i)$, which is composed of two terms to control the movement cost in deactivation phases and a term to control the movement cost in growth phases. Having two different terms for deactivation phases stems from two different ways of bounding the cost in these phases, depending on whether $x_\ell$ is larger or smaller than $1/w^2$.

\subsection{Explicit algorithm and basic lemmas}
\label{sec: explicit_and_basic}
In this section, we establish some fundamental properties of the algorithm. The proofs of these lemmas are deferred to Appendix~\ref{ap: proof_of_explicit_and_basic}.

The following lemma (specialized with $\bgamma=0$) explicitly describes the configuration of the algorithm at any given layer by specifying the probabilities $\frac{x_u}{x_{p(u)}}$ of being in the subtree of $u$ conditioned on being in the subtree of $p(u)$. 

\begin{lemma}[Explicit algorithm]\label{lem:explicitAlgo}
    Let $\bgamma\in\R^\Lcal$ and
    \begin{align*}
        \bx=\argmin_{\bx\in K}\sum_{u\in V\setminus\{r\}}h_{p(u)}x_u\log\frac{x_u}{x_{p(u)}} + \sum_{\ell\in\Lcal}\gamma_\ell x_\ell.
    \end{align*}
    Then for each $u\in V\setminus\{r\}$,
    \begin{align*}
        \frac{x_u}{x_{p(u)}} = \frac{y_u}{\sum\limits_{v\in c(p(u))}y_v},
    \end{align*}
    where $y_u$ is defined recursively via
    \begin{align*}
        y_u=\begin{cases}
            e^{-\gamma_u/h_{p(u)}}&\text{if }u\in\Lcal\\
            \left(\sum\limits_{v\in c(u)} y_v\right)^{h_u/h_{p(u)}}\quad&\text{if }u\notin\Lcal.
        \end{cases}
    \end{align*}
    In particular, if $\bgamma=0$, then $1\le y_u\le |\Lcal_u|$.
\end{lemma}

As we will consider continuously evolving configurations, it will be useful to understand their dynamics. These are expressed most naturally by differential equations describing the evolution of the conditional probabilities $\frac{x_u}{x_{p(u)}}$. The premise~\eqref{eq:MDcontainment} of the following lemma will be satisfied for different functions $\bgamma$ in deactivation and growth steps.

\begin{lemma}[Continuous dynamics]\label{lemma: dynamics}
    Let $\bgamma\colon [0,\infty)\to \R^\Lcal$ and $x\colon[0,\infty)\to K$ be continuously differentiable and satisfy $x_u(t)>0$ for all $u\in V$, $t\in[0,\infty)$ and
\begin{equation}
    \nabla^2 \Phi(\bx(t))\dot\bx(t) +\sum_{\ell \in \Lcal}\gamma_\ell(t)\be_\ell  \in -N_{K}\label{eq:MDcontainment}
\end{equation}
for all $t\in[0,\infty)$. Then for each node $u\in V\setminus\{r\}$,
    \begin{equation*}
        \partial_t\frac{x_u}{x_{p(u)}} = \frac{1}{h_{p(u)}x_{p(u)}}\left( \frac{x_u}{x_{p(u)}}\sum_{\ell \in \Lcal_{p(u)}}\!\!\!x_\ell\gamma_\ell \quad -\quad \sum_{\ell \in  \Lcal_u} x_\ell\gamma_\ell\right).
    \end{equation*}
\end{lemma}

The next lemma bounds the movement cost in terms of the change of the conditional probabilities.

\begin{lemma}[Movement cost]\label{lemma: mvt_cost}
    The instantaneous cost during continuous movement is at most
    \begin{equation*}
        2\sum_{u\in V \setminus \{r\}} h_{p(u)}x_{p(u)}\left(\partial_t \frac{x_u}{x_{p(u)}}\right)_-,
    \end{equation*}
    where we write $(z)_-=\max\{-z,0\}$.
\end{lemma}
\subsection{Deactivation phase}\label{sec:deactivate}
This section examines the deactivation of leaves in $\Lcal(i)$ at a given step $i$. For simplicity, and without loss of generality\footnote{The deactivation of multiple leaves can take place by iterating this method with all elements of $\Lcal(i)\setminus V(i+1)$. Enumerating deactivated leaves in some arbitrary order $\ell_1, \dots, \ell_j$, one has $|\bar V(i+1)|-|\bar V(i)| = d_{\ell_1}+\dots+d_{\ell_j}$.}, we focus on the case where only one leaf $\ell$ is deactivated and we  denote by $d_\ell\geq 1$ the corresponding number of deactivated nodes, i.e., which had $\ell$ as their unique descendant in $\Lcal(i)$. Denoting by $\bx'$ the configuration after the deactivations, the goal of this section is to show that
\begin{equation}\label{eq:changePotDel}
    \OT(\bx(i),\bx')\leq 4d_\ell/w + 4\log w\cdot \left(\Phi_i(\bx')-\Phi_i(\bx(i))\right), \tag{\ref{eq:main-delete}'}
\end{equation}
where the configuration after deactivations satisfies
\begin{equation*}
\bx' = \argmin_{\bx\in K(i) \text{ and } x_\ell = 0} \Phi_i(\bx).
\end{equation*}

For $t\in[0,\infty)$, we define $\bx(t)$ to interpolate between $\bx(i)$ and $\bx'$ as follows,
\begin{equation*}
    \bx(t) = \argmin_{\bx \in K(i)} \Phi_i(\bx)+t x_\ell.
\end{equation*}
We start by stating the following claims on $\bx(t)$: 
\begin{lemma}
    The following statements are satisfied:
    \begin{enumerate}[label=(\alph*)]
    \item $\bx(t)$ satisfies the explicit form given in Lemma \ref{lem:explicitAlgo} with $\gamma(t)=t\be_\ell$. Therefore, it is differentiable, and $\bx(t)$ converges to $\bx'$.\label{claim:explicit}
    \item $\bx(t)$ satisfies the mirror descent equation $\nabla^2\Phi_i(\bx(t)) \dot \bx(t) \in - \be_\ell - N_{K(i)}$ and thus, the dynamics given by Lemma \ref{lemma: dynamics}. Also, $\partial_t D(\bx',\bx(t)) = - x_{\ell}(t)$ and $\dot x_{\ell}(t) \leq - \frac{1}{2d_\ell}x_\ell(t)$. \label{claim:mirror}   
\end{enumerate}
\end{lemma}

\begin{proof}
We provide a short proof of the above claims.

\textit{\ref{claim:explicit} Explicit form.} $\bx(t)$ clearly satisfies the conditions of Lemma \ref{lem:explicitAlgo}. Since the expression only involves differentiable operations, $\bx(t)$ is differentiable. Convergence to $\bx'$ holds because the explicit form of $\bx(t)$ converges to the explicit form of $\bx'$ (which is given by invoking Lemma \ref{lem:explicitAlgo} with $\gamma=0$ and the polytope $K$ where the leaf $\ell$ is removed).

\textit{\ref{claim:mirror} Mirror descent.} For any $t\geq 0$, the optimality of $\bx(t)$ implies $\nabla \Phi_i(\bx(t))+t\be_\ell \in -N_{K(i)}$. Taking the derivative with respect to $t$, we obtain, $\nabla^2 \Phi_i(\bx(t))\dot\bx(t) \in -\be_\ell-N_{K(i)}$. By definition of the Bregman divergence, $\partial_t D(\bx',\bx(t)) = \dot \bx(t)^T\nabla^2\Phi_i(\bx(t))(\bx(t)-\bx') = -x_\ell(t)$, where we used that $x_\ell' = 0$ and that $\bx(t)-\bx'$ is orthogonal to $N_{K(i)}$. Then, applying Lemma \ref{lemma: dynamics} to the lowest node $u$ on the path $\ell\to r$ that has an active sibling, and observing that $(x_\ell(t),\dot x_\ell(t)) = (x_u(t),\dot x_u(t))$ and that $h_{p(u)}=d_\ell$, we have $d_\ell\left(\frac{\dot x_{\ell}}{x_{\ell}} - \frac{\dot x_{p(u)}}{x_{p(u)}}\right) = \frac{x_\ell}{x_{p(u)}}-1$. Then, by Lemma~\ref{lem:explicitAlgo}, we have $\frac{x_\ell}{x_{p(u)}} \leq 1/2$  and $\dot x_{p(u)} \leq 0$, since all conditional probabilities on $\ell\to r$ are non-increasing. Therefore, $d_\ell \frac{\dot x_\ell}{x_\ell}\leq -1/2$.
\end{proof}

Using claim \ref{claim:mirror} above, we note that
\begin{equation}
    \int_{t=0}^\infty x_\ell(t)dt =  -\int_{t=0}^\infty \partial_t D(\bx',\bx(t))dt = D(\bx',\bx(i)) =\Phi_i(\bx')-\Phi_i(\bx(i))  \label{eq: breg-potential}
\end{equation}
where the last equation uses that $\nabla \Phi_i(\bx(i)) \in -N_{K(i)}$ (by optimality of $\bx(i)$) is orthogonal to $\bx'-\bx(i)$.

For the proof deactivation step \eqref{eq:changePotDel}, in light of \eqref{eq: breg-potential}, it suffices to prove that 
\begin{equation}\label{eq: new_deletion_goal}
    \OT(\bx(i),\bx')  \leq4d_\ell/w +4 \log w\int_{t=0}^\infty x_\ell(t) dt. 
\end{equation}
Since $x_\ell(t)$ is decreasing during the deletion of $\ell$, if $x_\ell(i)>1/w^2$, there exists an instant $\bar t> 0$ such that $x_\ell(t) \geq 1/w^2$ for $t\leq \bar t$ and  $x_\ell(t) \leq  1/w^2$ for $t\geq \bar t$. If $x_\ell(i)\leq 1/w^2$, we simply set $\bar t = 0$. By the triangle inequality, $\OT(\bx(i),\bx')\le \int_{t=0}^{\infty}  \partial_t \Cost(t)dt$. 
To achieve the above goal \eqref{eq: new_deletion_goal}, it will thus suffice to prove the two following equations:
\begin{align}
    \int_{t=0}^{\bar t}  \partial_t \Cost(t)dt &\leq 4 \log w \int_{t=0}^{\bar t} x_\ell(t)dt,\label{eq: subgoal1}\\
    \int_{\bar t}^{\infty}  \partial_t \Cost(t)dt &\leq 4d_\ell /w\label{eq: subgoal2}
\end{align}

Applying Lemma~\ref{lemma: mvt_cost} and Lemma~\ref{lemma: dynamics} with $\gamma = \be_\ell$, the instantaneous cost during deactivation of leaf $\ell$ satisfies
\begin{align}
    \partial_t \Cost(t) &\leq 2x_\ell \sum_{u\in \ell \rightarrow r} \left(1-\frac{x_u}{x_{p(u)}}\right).\label{eq: costphase2}
\end{align}
In particular, this shows that $\partial_t \Cost(t) \leq 2w x_\ell$, because $w$ bounds the number of non-zero terms in the above sum. Together with the claim \ref{claim:mirror} that $x_\ell \leq -2d_\ell \dot x_\ell$, we get that
\begin{equation*}\int_{\bar t}^{\infty}  \partial_t \Cost(t)dt \leq -4d_\ell w\int_{t = \bar t}^\infty  \dot x_\ell(t) dt = 4d_\ell w x_\ell(\bar t) \leq 4d_\ell /w,
\end{equation*}
which proves inequality \eqref{eq: subgoal2}.

To conclude, we now need to prove \eqref{eq: subgoal1}.  Since $1-z\le -\log z$ for all $z>0$,
\begin{equation*}
    \sum_{u\in \ell \rightarrow r} \left(1-\frac{x_u}{x_{p(u)}}\right) \le -\log\left(\prod_{u\in \ell \rightarrow r}\frac{x_u}{x_{p(u)}} \right) = -\log x_\ell.
\end{equation*}
For $t< \bar t$, since $x_\ell(t) \geq 1/w^2$, we can further bound this by $2\log w$. In turn, using \eqref{eq: costphase2}, we get that
\begin{equation*}
    \partial_t \Cost(t) \leq 4\log w \cdot x_\ell,
\end{equation*}
which completes the proof of inequality \eqref{eq: subgoal1}.

\subsection{Growth phase}\label{sec:grow}
In this section, we will account for the growth of the tree by interpolating continuously between the configuration defined after all deactivations, $\bx'$, and the configuration defined at the next layer $\bx(i+1)$. We consider the following trajectory, where we now have for $t\in (0,1]$,
\begin{equation*}
    \bx(t) = \argmin_{\bx \in K(i+1)} \Phi_t(\bx) \quad \text{where} \quad \Phi_t(\bx) = \Phi_{i}(\bx)+t\sum_{\ell \in \Lcal(i+1)}x_\ell \log x_\ell.
\end{equation*}
Taking the derivative with respect to $t$ in $\Phi_t(\bx(t))$ yields, by the optimality of $\bx(t)$, that 
$\partial_t \Phi_t(\bx(t)) = \sum_{\ell \in \mathcal{L}(i+1)} x_\ell \log x_\ell \geq -\log w$. Hence, $\Phi_{i+1}(\bx(i+1))-\Phi_i(\bx')=\int_0^1\partial_t \Phi_t(\bx(t)) dt\ge -\log w$ yields inequality \eqref{eq:main-elong-2}.
To complete the analysis of the growth step we must therefore prove inequality \eqref{eq:main-elong-1} stating that 
\begin{equation}
\OT(\bx',\bx(i+1))\leq 2(1+\log w)^2+1.\tag{\ref{eq:main-elong-1}}\label{eq:main-elong-1Repeated}
\end{equation} 

Though the configuration $\bx(t)$ is a member of $K(i+1)$, we will not pay for movements between layer $i$ and layer $i+1$ until the end of the growth phase.\footnote{More precisely, we initially pay only for the movement costs associated to the projection of $\bx(t)$ onto $K(i)$.} This is because we can effectively move the probability mass between points of $\Lcal(i)$, and then move the mass from $\Lcal(i)$ to $\Lcal(i+1)$ at the very end of the elongation. This final movement is paid by the $+1$ term in inequality \eqref{eq:main-elong-1Repeated}.
Thus, to preserve \eqref{eq:main-elong-1Repeated}, it suffices to show that the instantaneous movement cost, for $t\in [0,1]$, is bounded as
\begin{align}
    \partial_t \Cost(t)\le 2(1+\log w)^2.\label{eq:growthToShow}
\end{align}
Differentiating the optimality condition $\nabla \Phi_t(\bx(t))\in- N_{K(i+1)}$, we get that $\bx(t)$ satisfies $\nabla^2 \Phi_t(\bx) \dot \bx +\sum_{\ell \in \Lcal(i+1)} (1+\log x_\ell)\be_\ell\in  - N_{K(i+1)}$. 
Since $\sum_{\ell \in \Lcal(i+1)} \be_\ell \in N_{K(i+1)}$, we can rewrite this expression as
\begin{equation*}
    \nabla^2 \Phi_t(\bx) \dot \bx + \sum_{\ell \in \Lcal(i+1)} \log x_\ell\cdot \be_\ell \in - N_{K(i+1)}.
\end{equation*}
In light of Lemma \ref{lemma: dynamics} we have for all $u\in V(i)$, 
\begin{equation*}
    \partial_t \frac{x_u}{x_{p(u)}} = \frac{1}{h_{p(u)}x_{p(u)}}\left( \frac{x_u}{x_{p(u)}}\sum_{\ell \in\Lcal_{p(u)}(i+1)}x_\ell\log x_\ell -\sum_{\ell \in \Lcal_u(i+1)} x_\ell \log x_\ell\right),
\end{equation*}
where we define $h_{p(u)} = h_{p(u)}(i)+t$. By Lemma \ref{lemma: mvt_cost}, we recall that the continuous movement is bounded by
\begin{align}
   \partial_t \Cost(t) &\leq 2\sum_{u\in V(i)\setminus \{r\}}h_{p(u)}x_{p(u)}\left(\partial_t \frac{x_u}{x_{p(u)}}\right)_- \nonumber \\
   &= 2\sum_{u\in V(i)\setminus \{r\}} \left(\frac{x_u}{x_{p(u)}}\sum_{\ell \in\Lcal_{p(u)}(i+1)}x_\ell\log x_\ell -\sum_{\ell \in \Lcal_u(i+1)} x_\ell \log x_\ell \right)_-.\label{eq:growthCostInitialBound}
\end{align}

For a node $p\in V(i)\setminus\Lcal(i)$, let $c^-(p)\subset c(p)$ denote the set of children $u$ of $p$ that have a non-zero contribution to this sum, i.e., satisfying $\frac{x_u}{x_{p}}\sum_{\ell \in\Lcal_{p}}x_\ell\log x_\ell -\sum_{\ell \in \Lcal_u} x_\ell \log x_\ell<0$, and let $c^+(p)=c(p)\setminus c^-(p)$. Let $\Lcal_p^+ = \bigcup_{u\in c^+(p)}\Lcal_u(i+1)$, $w_p^+=|\Lcal_p^+|$ and $x_p^+ = \sum_{u\in c^+(p)}x_u = \sum_{\ell\in\Lcal_p^+}x_\ell$, and define $\Lcal_p^-$, $w_p^-$ and $x_p^-$ similarly, so that $x_p^++x_p^-=x_p$ and $\Lcal_p^+\cup\Lcal_p^-=\Lcal_p(i+1)$. Then
\begin{align}
    \sum_{u\in c(p)} \left(\frac{x_u}{x_{p}}\sum_{\ell \in\Lcal_{p}(i+1)}x_\ell\log x_\ell -\sum_{\ell \in \Lcal_u(i+1)} x_\ell \log x_\ell \right)_-\notag
   &= \sum_{\ell \in \Lcal_p^-} x_\ell \log x_\ell - \frac{x_p^-}{x_{p}}\sum_{\ell \in\Lcal_{p}(i+1)}x_\ell\log x_\ell\notag\\
   &= \frac{x_p^+}{x_{p}}\sum_{\ell \in\Lcal_{p}^-}x_\ell\log x_\ell - \frac{x_p^-}{x_{p}}\sum_{\ell \in \Lcal_p^+} x_\ell \log x_\ell\notag\\
   &\le \frac{x_p^+x_p^-}{x_{p}}\log x_p^- - \frac{x_p^-x_p^+}{x_{p}}\log \frac{x_p^+}{w_p^+}\notag\\
   &\le \frac{x_p^+x_p^-}{x_p}\log(w_p^+w_p^-).\label{eq:growthIntermediateBound}
\end{align}
where the first inequality uses that $\sum x_\ell\log x_\ell$ subject to the constraints $x\ge 0$ and $\sum x_\ell=x_p^-$ (resp. $\sum x_\ell=x_p^+$) is maximized (resp. minimized) when all mass is concentrated on a single leaf (resp. spread equally across the $w_p^+$ leaves), and the last inequality follows from the (yet unproven) relation
\begin{align}
    x_p^-\le x_p^+w_p^-.\label{eq:claimedIneq}
\end{align}

To see that \eqref{eq:claimedIneq} is true, note that by Lemma~\ref{lem:explicitAlgo}, there exist $1\le y_u\le |\Lcal_u|$ such that
\begin{align*}
    \frac{x_p^+}{x_p} = \frac{\sum_{u\in c^+(p)}y_u}{\sum_{u\in c(p)}y_u} \ge \frac{1}{1+\sum_{u\in c^-(p)}|\Lcal_u|} = \frac{1}{1+w_p^-}.
\end{align*}
Applying this twice, we get
\begin{align*}
    \frac{x_p^+w_p^-}{x_p} \ge \frac{w_p^-}{1+w_p^-} = 1 - \frac{1}{1+w_p^-} \ge 1 - \frac{x_p^+}{x_p} = \frac{x_p^-}{x_p},
\end{align*}
and multiplying by $x_p$ yields \eqref{eq:claimedIneq}.

Combining~\eqref{eq:growthCostInitialBound} and \eqref{eq:growthIntermediateBound} yields
\begin{align*}
    \partial_t\Cost(t) \le 2\sum_{p\in V(i)\setminus\Lcal(i)} \frac{x_p^+x_p^-}{x_p}\log(w_p^+w_p^-).
\end{align*}

Finally, our objective \eqref{eq:growthToShow} is implied by invoking the following lemma with $u=r$.

\begin{lemma}\label{lem: induction_for_growth}
    For each $u\in V(i)$, writing $V_u\subseteq V(i)$ for the nodes in the subtree rooted at $u$ and $w_u=|\Lcal_u(i+1)|$, it holds that
    \begin{align*}
        \sum_{p\in V_u\setminus\Lcal(i)} \frac{x_{p}^+x_p^-}{x_{p}} \log(w_p^+ w_p^-)\le x_u (1+\log w_u)^2.
    \end{align*}
\end{lemma}

\begin{proof}
    We proceed by induction on the level of $u$. If $u\in\Lcal(i)$, the inequality follows trivially since the left-hand side is an empty sum.

    For the induction step, consider some $u\in V(i)\setminus\Lcal(i)$. Then 
    \begin{align*}
        \sum_{p\in V_u\setminus\Lcal(i)} \frac{x_{p}^+x_p^-}{x_{p}} \log(w_p^+ w_p^-) &= \frac{x_{u}^+x_u^-}{x_{u}} \log(w_u^+ w_u^-) + \sum_{v\in c(u)}\sum_{p\in V_v\setminus\Lcal(i)} \frac{x_{p}^+x_p^-}{x_{p}} \log(w_p^+ w_p^-) \\
        &\le \frac{x_{u}^+x_u^-}{x_{u}} \log(w_u^+ w_u^-) + \sum_{v\in c(u)}x_v(1+\log w_v)^2\\
        &\le \frac{x_{u}^+x_u^-}{x_{u}} \log(w_u^+ w_u^-) + x_u^+(1+\log w_u^+)^2 + x_u^-(1+\log w_u^-)^2
    \end{align*}
    where the first inequality follows from the induction hypothesis. We may assume $w_u^-\ge 1$ and $w_u^+\ge 1$, since otherwise $x_u^+$ or $x_u^-$ is $0$ and the lemma follows directly from the induction hypothesis applied to the children of $u$.
    Dividing by $x_u$ and writing $q:=\frac{x_u^+}{x_u}$, so that $1-q=\frac{x_u^-}{x_u}$, it remains to show that
    \begin{align}
    q(1-q)\log(w_u^+w_u^-) + q(1+\log w_u^+)^2 + (1-q)(1+\log w_u^-)^2 &\le (1+\log w_u)^2\label{eq:indStepToShow}
    \end{align}
    for all $q\in[0,1]$.

    By concavity of $w\mapsto(1+\log w)^2$ for $w\ge 1$, we have
    \begin{align*}
        q(1+\log w_u^+)^2 + (1-q)(1+\log w_u^-)^2 \le (1+\log(qw_u^++(1-q)w_u^-))^2.
    \end{align*}
    By symmetry, assume $w_u^+\ge w_u^-$, so we can write $w_u^+=(1+a)w_u^-$ for some $a\ge 0$. Using $w_u=w_u^++w_u^-$, and writing $d:=ew_u^-$ (where $e$ is the base of $\log$), inequality~\eqref{eq:indStepToShow} follows if we show that
    \begin{align*}
    q(1-q)\log((1+a)d^2) + \log^2((1+aq)d) &\le \log^2 ((2+a)d).
    \end{align*}
    for all $q\in[0,1]$, $a\ge 0$ and $d\ge e$. Indeed,
    \begin{align}
        &\phantom{=} q(1-q)\log((1+a)d^2) + \log^2((1+aq)d) - \log^2 ((2+a)d)\notag\\
        &= q(1-q)(2\log d + \log(1+a)) + 2\log d\cdot\log\frac{1+aq}{2+a} + \log^2(1+aq) - \log^2(2+a)\notag\\
        &\le \log(2+a)\log\frac{1+a}{2+a} + \log(1+aq)\log\frac{1+aq}{1+a}\label{eq:logfrac}\\
        &\le 0\notag,
    \end{align}
    where inequality~\eqref{eq:logfrac} uses
    \begin{align*}
        q(1-q)\le \min\left\{\log 2,\log \frac{1}{q}\right\}\le \log\frac{2+a}{1+aq}.
    \end{align*}
\end{proof}

\paragraph*{Acknowledgements.} The authors thank the anonymous reviewers for their helpful comments. RC thanks Laurent Massoulié for insightful discussions. The work was carried out while XB was affiliated with the University of Oxford.

\bibliographystyle{alpha}
\bibliography{biblio}

\appendix
\section{Simple results for unweighted layered graph traversal}\label{ap:simpleResults}
This section discusses a few simple algorithms and bounds for traversing unweighted layered graphs. 
First, we give a tight bound on deterministic algorithms; then, we give a randomized lower bound and discuss two elementary randomized algorithms that fail to provide substantial competitive guarantees.

\begin{proposition}[Deterministic Tight Bound]
For every $w\in \Nbb$, the competitive ratio of deterministic unweighted layered graph traversal with width $w$ is exactly $2w-1$.
\end{proposition}

The upper bound is achieved by a simple depth-first search. This algorithm uses the edges on the path to the target once, and other edges at most twice. Since the number of edges traversed is at most $w$ per layer, the competitive ratio is $2w-1$.

For the lower bound, consider a star-shaped graph $G_{\text{star}}$ comprising a root $r$ and $w$ chains starting from $r$, as illustrated by Figure~\ref{fig:star}. 
A deterministic algorithm must sequentially explore each chain's endpoint to find the chain with depth $t$. 
For any $\Acal$ and any depth $i$, we assign length $i - w + j$ to the $j$-th chain $\Acal$ would explore. 
This graph instance would make $\Acal$ traverse all $w$ chains, incurring a total cost of $(2w - 1)i-O(w^2)$. Letting $i\to\infty$ shows that no deterministic algorithm can have a competitive ratio less than $2w-1$.

\begin{figure}[htbp]
\centering
\begin{subfigure}[b]{0.30\textwidth}
    \centering
    \includegraphics[width=0.6\textwidth]{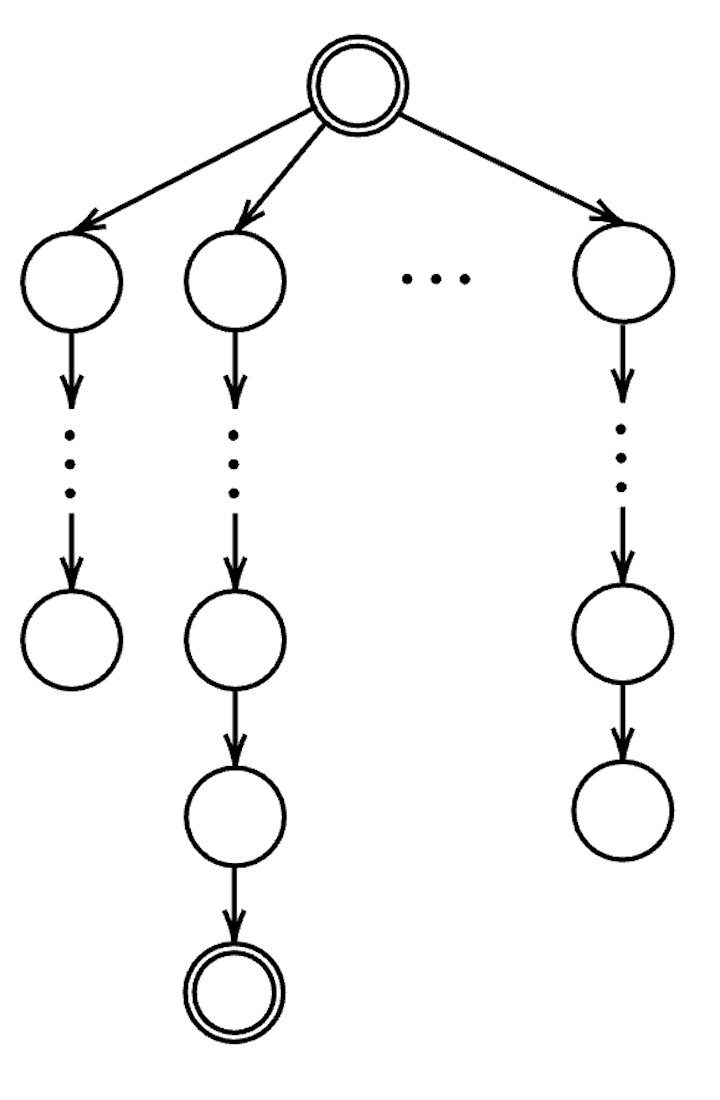}
    \caption{Star Graph}
    \label{fig:star}
\end{subfigure}
\hfill
\begin{subfigure}[b]{0.30\textwidth}
    \centering
    \includegraphics[width=0.6\textwidth]{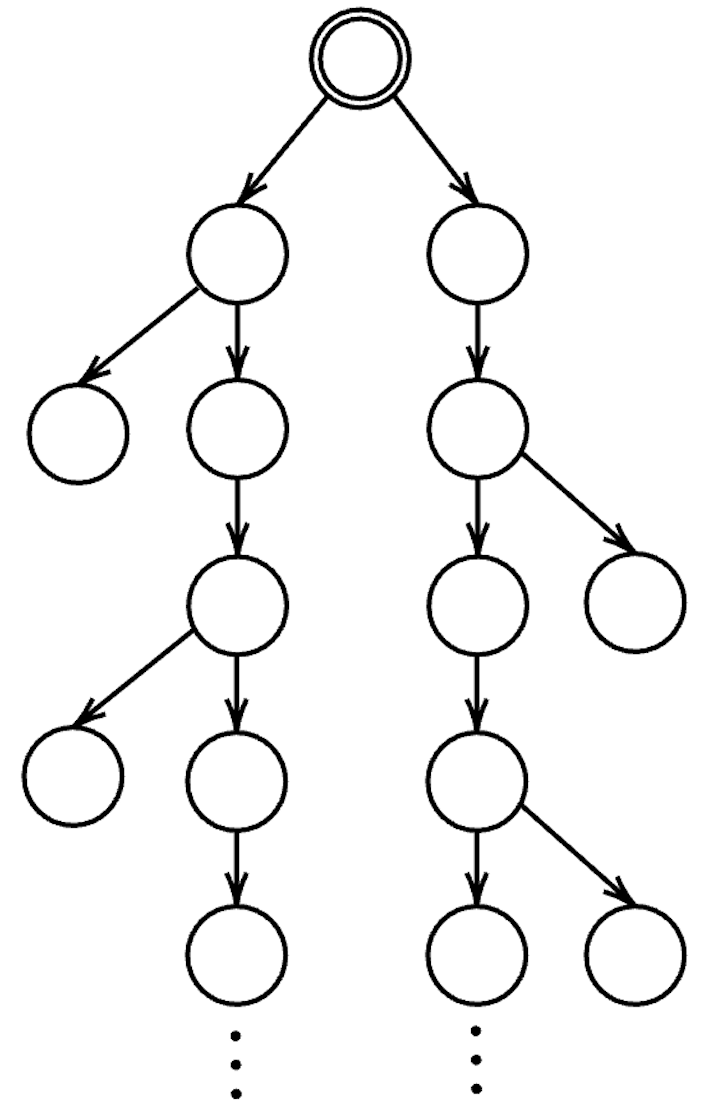}
    \caption{Alternating Branch Graph}
    \label{fig:branches}
\end{subfigure}
\hfill
\begin{subfigure}[b]{0.30\textwidth}
    \centering
    \includegraphics[width=0.85\textwidth]{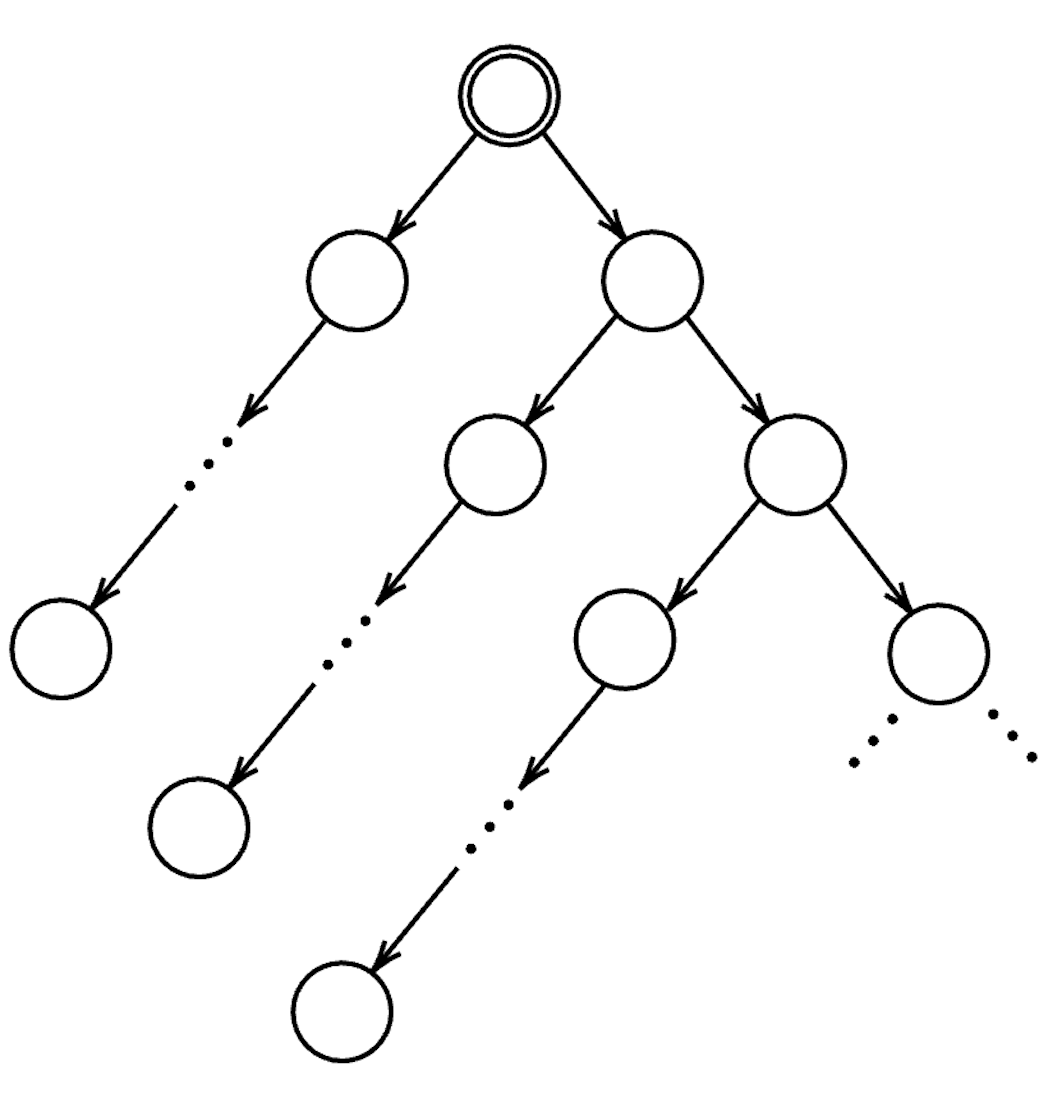}
    \caption{"Comb" Graph}
    \label{fig:comb}
\end{subfigure}
\caption{Three Instances of Unweighted Layered Graph.}
\label{fig:graphs-comparison}
\end{figure}

\begin{proposition}[Randomized Lower Bound]
Any randomized algorithm for unweighted layered graph traversal has a competitive ratio of $\Omega(\log w)$.
\end{proposition}
The proof is analogous to \cite[Theorem 12]{fiat1998competitive}. For any randomized algorithm, we consider its fractional counterpart $\Acal$. 
We apply $\Acal$ on the star-shaped graph $G_{\text{star}}$ as in Figure~\ref{fig:star}. Initially, each chain has length $i - w + 1$ for sufficiently large $i$.
At layer $i - w + 1$, at least one endpoint is occupied with probability at least $\frac{1}{w}$. We give that node no children in the next layer and extend the length of the remaining chains. We repeat this process on layers $i - w + 2$ to $i$.
Then, the expected cost incurred by $\Acal$ is $\Omega(i)$ times the sum, over each chain, of the probability that $\Acal$ visits its endpoint, which is $\sum_{j=1}^{w} \frac{1}{j} = \Omega(\log w)$.

To find efficient traversal algorithms, we examined various elementary ideas. Among them, none have demonstrated a competitive ratio below $\Omega(w)$, failing to improve on the deterministic DFS.

\paragraph{Random DFS Algorithm.}
Perhaps the most immediate (randomized) generalization of depth-first search is to consider the algorithm where the agent chooses the branch that it will explore next uniformly at random among the closest active branches, instead of always choosing one branch deterministically, which can be exploited by the adversary (as we did in the study of deterministic depth-first search). It is easy to see that the fractional perspective on this algorithm actually corresponds at all times to the configuration $\bx(i)$ that is defined by induction by $x_r =1$ and $\forall u \in V\setminus \{r\},~ x_u = x_{p(u)}/|c(p(u))|$ where $|c(p(u))|$ denotes the number of active children of $u$ at round $i$. We note that this method is $\Omega(w)$ competitive when applied to the ``Comb'' graph illustrated in Figure~\ref{fig:comb}. 
In this graph, each left-extending chain from the rightmost branch has length $w$, and the algorithm explores each dead end with a probability of one-half. 
Consequently, the total expected cost for this algorithm is $\Omega(w i)$ at the $i$-th layer.

A slightly more elaborate lower bound shows that the algorithm that moves, upon reaching a dead-end, to a random node that is at most some constant factor $L$ (possibly depending on $w$) further away than the closest active node in the new layer, also has competitive ratio $\Omega(w)$.

\paragraph{Uniform Algorithm.} The problem with the random depth-first search algorithm seems to be that it lets too much probability mass concentrate on a single leaf, and that the adversary is then encouraged to delete that precise leaf. One intuitive idea to overcome this issue is to consider the algorithm that maintains a uniform probability distribution on the active leaves instead (moving from one layer to the next using an optimal transport coupling). 
This algorithm fails in the case of a two-branch tree, as depicted in Figure~\ref{fig:branches}. 
The design of this tree causes the ultimate layer to alternate between having \{2, 1\} and \{1, 2\} nodes in both of its branches, incurring a movement cost of $\Omega(i)$ for each step due to the constant probability of switching branches. Hence, this algorithm has unbounded competitive ratio even for $w=3$.

\section{Equivalent expressions for 
\texorpdfstring{$\Phi$}{} 
}\label{ap: useful-identities}
The regularizer has two possible expressions, we give a detailed derivation here.
\begin{align*}
    \Phi(\bx) = \sum_{u\in V\setminus \{r\}}h_{p(u)}x_u\log\frac{x_u}{x_{p(u)}}
    &= \sum_{u\in V\setminus \{r\}}h_{p(u)}x_u\log x_u
    - \sum_{u\in V\setminus \{r\}}h_{p(u)}x_u\log x_{p(u)}, \nonumber\\ 
    &= \sum_{u\in V\setminus \{r\}}h_{p(u)}x_u\log x_u
    - \sum_{u\in V\setminus \Lcal}h_{u}x_u\log x_{u}, \nonumber\\
    \nonumber
    &= \sum_{u\in V\setminus \{r\}}(h_{p(u)}-h_u)x_u\log x_u,\\ 
    &=\sum_{u\in V} x_u\log x_u,
\end{align*}
where we used $\log x_r=\log1=0$ and $h_\ell=0$ for $\ell\in\Lcal$.

\section{Proof of the explicit algorithm and basic lemmas}
\label{ap: proof_of_explicit_and_basic}
In this section, we restate and prove the explicit algorithm (Lemma~\ref{lem:explicitAlgo}) and two fundamental lemmas regarding the continuous dynamics and movement cost (Lemma~\ref{lemma: dynamics} and \ref{lemma: mvt_cost}) of evolving configurations from Section~\ref{sec: explicit_and_basic}.

\begin{lemma}[Explicit algorithm]
    Let $\bgamma\in\R^\Lcal$ and
    \begin{align*}
       \bx=\argmin_{\bx\in K}\sum_{u\in V\setminus\{r\}}h_{p(u)}x_u\log\frac{x_u}{x_{p(u)}} + \sum_{\ell\in\Lcal}\gamma_\ell x_\ell.
    \end{align*}
    Then for each $u\in V\setminus\{r\}$,
    \begin{align}
        \frac{x_u}{x_{p(u)}} = \frac{y_u}{\sum\limits_{v\in c(p(u))}y_v},\label{eq: ap-explicit}
    \end{align}
    where $y_u$ is defined recursively via
    \begin{align*}
        y_u=\begin{cases}
            e^{-\gamma_u/h_{p(u)}}&\text{if }u\in\Lcal\\
            \left(\sum\limits_{v\in c(u)} y_v\right)^{h_u/h_{p(u)}}\quad&\text{if }u\notin\Lcal.
        \end{cases}
    \end{align*}
    In particular, if $\bgamma=0$, then $1\le y_u\le |\Lcal_u|$.
\end{lemma}
\begin{proof}
    Let
    \begin{align*}
        f(\bx)=\sum_{u\in V\setminus\{r\}}h_{p(u)}x_u\log\frac{x_u}{x_{p(u)}} + \sum_{\ell\in\Lcal}\gamma_\ell x_\ell.
    \end{align*}
    Since $f$ is strictly convex, $\bx$ is the unique point in $K$ satisfying the $\nabla f(\bx)\in -N_K(\bx)$. Thus, it suffices to show that the point $\bx$ defined via the ratios \eqref{eq: ap-explicit} satisfies these conditions. 
    
    Clearly, it satisfies $\bx\in K$. Moreover,
    \begin{align*}
        \partial_{x_u}f(\bx) &= h_{p(u)}\left(1+\log\frac{x_u}{x_{p(u)}}\right) - h_u\sum_{v\in c(u)}\frac{x_v}{x_u} + \gamma_u\one(u\in\Lcal)\\
        &= h_{p(u)}\left(1+\log y_u - \log\sum\limits_{v\in c(p(u))}y_v\right) - h_u + \gamma_u\one(u\in\Lcal)\\
        &= h_{p(u)}+h_{p(u)}\log y_u - h_{p(p(u))}\log y_{p(u)} - h_u + \gamma_u\one(u\in\Lcal).\\
    \end{align*}
Thus, letting $\lambda_u= h_u-h_{p(u)}\log y_u$ for $u\in V\setminus\Lcal$ and $\lambda_\ell=0=-\gamma_\ell+h_\ell-h_{p(\ell)}\log y_\ell$ for $\ell\in\Lcal$ (using $h_\ell=0$), we have
    \begin{align*}
        \nabla f(x) = \sum_{u\in V\setminus\{r\}}(\lambda_{p(u)}-\lambda_u)\be_u\in -N_K(x).
    \end{align*}
    If $\bgamma=0$, the inequalities that pertain to $y_u$ follow from the recursive definition by induction.
\end{proof}

\begin{lemma}[Continuous dynamics]
    Let $\bgamma\colon [0,\infty)\to \R^\Lcal$ and $x\colon[0,\infty)\to K$ be continuously differentiable and satisfy $x_u(t)>0$ for all $u\in V$, $t\in[0,\infty)$ and
\begin{equation}
    \nabla^2 \Phi(\bx(t))\dot\bx(t) +\sum_{\ell \in \Lcal}\gamma_\ell(t)\be_\ell  \in -N_{K} \label{eq: ap-MDcontainment}
\end{equation}
for all $t\in[0,\infty)$. Then for each node $u\in V\setminus\{r\}$,
    \begin{equation*}
        \partial_t\frac{x_u}{x_{p(u)}} = \frac{1}{h_{p(u)}x_{p(u)}}\left( \frac{x_u}{x_{p(u)}}\sum_{\ell \in \Lcal_{p(u)}}\!\!\!x_\ell\gamma_\ell \quad -\quad \sum_{\ell \in  \Lcal_u} x_\ell\gamma_\ell\right).
    \end{equation*}
\end{lemma}
\begin{proof}
    The first and second order derivatives of $\Phi(\bx)=\sum_{u\in V\setminus\{r\}}h_{p(u)}x_u\log\frac{x_u}{x_{p(u)}}$ are given\footnote{Note that the derivatives of the function $\sum_{u\in V}x_u\log x_u$ would be different, even though these functions are equal on $K$. The reason is that to calculate the derivatives, we need to consider $\Phi$ on an open neighborhood of $K$, such as $\R_+^V$, where the two definitions of $\Phi$ are no longer equivalent. For the same reason, replacing $\sum_{v\in c(u)}\frac{x_v}{x_u}$ by $1$ in the formula for the first-order derivative and calculating the second-order derivatives based on this would yield an incorrect Hessian. Both formulations yield the same dynamics though.} by
    \begin{align*}
        \partial_{x_u}\Phi(\bx) &= h_{p(u)}\left(1+\log\frac{x_u}{x_{p(u)}}\right) - h_u\sum_{v\in c(u)}\frac{x_v}{x_u}\\
        \partial_{x_ux_{u}}\Phi(\bx) &= \frac{h_{p(u)}}{x_u}+ h_u\sum_{v\in c(u)}\frac{x_v}{x_u^2}\\
        \partial_{x_ux_{p(u)}}\Phi(\bx) = \partial_{x_{p(u)}x_u}\Phi(\bx) &= -\frac{h_{p(u)}}{x_{p(u)}},
    \end{align*}
    with all other second-order derivatives being zero. Elements of the normal cone $N_K$ are of the form $\sum_{u\in V\setminus\{r\}}(\lambda_{u}-\lambda_{p(u)})\be_u$, where $\lambda_u\in\R$ is a Lagrange multiplier corresponding to the constraint $x_u= \sum_{v\in c(u)}x_v$ for $u\in V\setminus \Lcal$, and $\lambda_\ell=0$ for $\ell\in\Lcal$.
    Equation \eqref{eq: ap-MDcontainment} then becomes
    \begin{align*}
        \nabla^2 \Phi(\bx(t))\dot\bx(t) = - \sum_{\ell \in \Lcal}\gamma_\ell(t)\be_\ell  + \sum_{u\in V\setminus \{r\}}(\lambda_{p(u)}(t)-\lambda_u(t))\be_u
    \end{align*}
    for some $\lambda\colon[0,\infty)\to\R^V$.
    
    Extending the co-domain of $\bgamma$ from $\R^\Lcal$ to $\R^V$ by letting $\gamma_u(t)=\lambda_u(t)$ for $u\in V\setminus\Lcal$, and using $\lambda_\ell=0$ for $\ell\in\Lcal$ the row corresponding to $u$ can be written as
    \begin{align}
        \gamma_{p(u)}-\gamma_u &= \partial_{x_ux_u}\Phi(\bx)\dot x_u + \partial_{x_ux_{p(u)}}\Phi(\bx)\dot x_{p(u)} + \sum_{v\in c(u)}\partial_{x_ux_v}\Phi(\bx)\dot x_v\notag\\
        &= \left(\frac{h_{p(u)}}{x_u}+ h_u\sum_{v\in c(u)}\frac{x_v}{x_u^2}\right)\dot x_u - \frac{h_{p(u)}}{x_{p(u)}}\dot x_{p(u)} - \frac{h_u}{x_u}\sum_{v\in c(u)}\dot x_v\notag\\
        &= \frac{h_{p(u)}}{x_u}\dot x_u - \frac{h_{p(u)}}{x_{p(u)}}\dot x_{p(u)}\notag\\
        &= \frac{h_{p(u)}x_{p(u)}}{x_u}\partial_t\frac{x_u}{x_{p(u)}},\label{eq:gammaDiff}
    \end{align}
    where the penultimate equation uses $\sum_{v\in c(u)}x_v=x_u$ and $\sum_{v\in c(u)}\dot x_v=\dot x_u$ for $u\in V\setminus\Lcal$. To conclude the lemma, it suffices to show
    \begin{align*}
        \gamma_u = \sum_{\ell\in\Lcal_{u}}\frac{x_\ell\gamma_\ell}{x_u}.
    \end{align*}
    We show this by induction on the height of $u$. For $u\in \Lcal$ this is immediate. For $u\in V\setminus\Lcal$, it is implied by the following chain of equations, where we use \eqref{eq:gammaDiff} and the induction hypothesis:
    \begin{align*}
        0 = h_u\sum_{v\in c(u)}\partial_t\frac{x_v}{x_u}
        = \sum_{v\in c(u)}\frac{x_v}{x_{u}}(\gamma_u-\gamma_v) = \gamma_u - \sum_{v\in c(u)}\frac{x_v}{x_{u}}\sum_{\ell\in\Lcal_v}\frac{x_\ell\gamma_\ell}{x_v} = \gamma_u - \sum_{\ell\in\Lcal_{u}}\frac{x_\ell\gamma_\ell}{x_u}.
    \end{align*}
\end{proof}

\begin{lemma}[Movement cost]
    The instantaneous cost during continuous movement is at most
    \begin{equation*}
        2\sum_{u\in V \setminus \{r\}} h_{p(u)}x_{p(u)}\left(\partial_t \frac{x_u}{x_{p(u)}}\right)_-,
    \end{equation*}
    where we write $(z)_-=\max\{-z,0\}$.
\end{lemma}
\begin{proof}
Consider an infinitesimal time step, where the algorithm moves from configuration $\bx(t)$ to configuration $\bx(t+dt)$. Rather than going from $\bx(t)$ to $\bx(t+dt)$ directly, one could break the movement into pieces, where the algorithm goes through the vertices $u\in V\setminus\Lcal$ one by one (bottom-up, say) and only updates the ``conditional probabilities'' $\frac{x_v}{x_u}$ for $v\in c(u)$. Since these conditional probabilities fully specify $\bx$, the algorithm still reaches the same configuration $\bx(t+dt)$ in the end. Hence, the combined cost of these movements upper bounds the optimal transport cost from $\bx(t)$ to $\bx(t+dt)$.

Consider now a step where only the conditional probabilities of children of a single vertex $u$ are updated while other conditional probabilities are held constant. The instantaneous cost of this movement is
\begin{align*}
    2h_ux_u\sum_{v\in c(u)}\left(\partial_t\frac{x_v}{x_{u}}\right)_-
\end{align*}
since $x_u\left(\partial_t\frac{x_v}{x_{u}}\right)_-$ is the rate at which mass is moving from leaves in $\Lcal_v$ upwards to $u$ by a distance $h_u$, and by conservation of mass it is moving downwards the same distance to leaves below siblings of $v$. Summing over all $u\in V\setminus\Lcal$ yields the lemma.
\end{proof}

\end{document}